%% file: arXivmain.tex
\documentclass[conference,letterpaper]{IEEEtran}

\usepackage[cmex10]{amsmath} % Use the [cmex10] option to ensure complicance
\input{preamble}

\newif\ifextended
\newacronym{GSP}{GSP}{Graph signal processing}
\newacronym{TSP}{TSP}{Topological signal processing}
\newacronym{CSO}{CSO}{complexon shift operator}
\newacronym{SCSP}{SCSP}{simplicial complex signal processing}
\newacronym{CSP}{CSP}{complexon signal processing}
\newacronym{GSO}{GSO}{graph shift operator}
\newacronym{GRSO}{GRSO}{graphon shift operator}
\newacronym{GRSP}{GRSP}{graphon signal processing}
\newacronym{LSI}{LSI}{linear shift-invariant}

\begin{document}
\title{Spectral Convergence of Complexon Shift Operators} 

%%% Single author, or several authors with same affiliation:
\author{%
\IEEEauthorblockN{Purui~Zhang, Xingchao~Jian, Feng~Ji, Wee~Peng~Tay, and Bihan~Wen}
\IEEEauthorblockA{School of Electrical and Electronic Engineering\\
Nanyang Technological University, Singapore\\
Email: \{PURUI001, XINGCHAO001\}@e.ntu.edu.sg, \{jifeng, wptay, bihan.wen\}@ntu.edu.sg}
}

\maketitle

%%%%%%
%% Abstract: 
%% If your paper is eligible for the student paper award, please add
%% the comment "THIS PAPER IS ELIGIBLE FOR THE STUDENT PAPER
%% AWARD." as a first line in the abstract. 
%% For the final version of the accepted paper, please do not forget
%% to remove this comment!
%%

\begin{abstract}
\gls{TSP} utilizes simplicial complexes to model structures with higher order than vertices and edges.
In this paper, we study the transferability of TSP via a generalized higher-order version of graphon, known as complexon. 
We recall the notion of a complexon as the limit of a simplicial complex sequence \cite{LDSC}. 
Inspired by the graphon shift operator and message-passing neural network, we construct a marginal complexon and complexon shift operator (CSO) according to components of all possible dimensions from the complexon.
We investigate the CSO's eigenvalues and eigenvectors and relate them to a new family of weighted adjacency matrices. We prove that when a simplicial complex signal sequence converges to a complexon signal, the eigenvalues, eigenspaces, and Fourier transform of the corresponding CSOs converge to that of the limit complexon signal. This conclusion is further verified by two numerical experiments. These results hint at learning transferability on large simplicial complexes or simplicial complex sequences, which generalize the graphon signal processing framework.
\end{abstract}

\section{Introduction}
\label{sec:intro}

%\gls{GSP} provides powerful utensils to process signals supported on a graph. 
\gls{GSP} offers powerful tools for modeling signals associated with graph structures\cite{OrtPie:J18}.
When presented with a fixed graph framework, one can design graph filters\cite{SanJos:C13,RuiAle:J20} and graph neural networks\cite{Sca:J09,Gama:J19} tailored for diverse tasks, including regression and classification\cite{DonFro:J20,LiaWan:J22}, in which the eigendecomposition of the graph filter \gls{GSO} plays a pivotal role\cite{SegAle:J17}. There are mainly two extensions for \gls{GSP}, one from the aspect of higher-order geometric structures\cite{ZhaCui:J20} and another from the aspect of asymptotic analysis\cite{MorLeu:J21}.

The first extension addresses the limitations of a graph structure \cite{ZhaCui:J20, JiaJiTay23}. Since a graph only captures information on nodes and edges, it cannot represent higher-order relationships between multiple nodes. One approach is to use hypergraphs to model the higher-order relationships \cite{ZhaXia:C22,PenArc:J23,JiaJiTay23}. However, in some cases, signals are embedded in a specific topological structure, such as a manifold\cite{WanAle:C22}. A simplicial complex becomes a more appropriate model since it can represent data on a structure with the help of homology\cite{BarSar:J20,BatCla:C23,Yang:C22,LeeJiTay22,WuYip:J24}. To develop \gls{SCSP}, Hodge Laplacians are the major components to be used to derive generalized Laplacians\cite{JiTay:A22} as suitable shift operators. Furthermore, it is possible to consider generalized signals on the graph vertices \cite{JiTay19}. 

However, even with the first extension, the dynamic and large-scale structures encountered in signal processing turn out to be a problem. As for standard \gls{GSP} and \gls{SCSP}, the topological structure is assumed to be fixed. When the structure itself varies, signal processing elements like shift operators, filters, and Fourier transforms, also change \cite{JiTayOrt23}. Besides, signal processing techniques such as Fourier transform are usually prohibitively expensive on large graphs or simplicial complexes. 

Hence, the second extension explores the limit structure of signal processing in order to deal with dynamic and large-scale structures. The papers \cite{RuiChaRib:J21,RuiCha:C20,Mas23} utilize graphons to study the transferability of graph filters. Graphon signal processing tools such as graphon shift operators and graphon Fourier transforms are introduced to investigate the transferability of a graphon as the limit of a graph sequence. Analogous to a graphon, a complexon is defined as the limit of simplicial complexes and the sampling of large simplicial complex structures \cite{LDSC}. However, signal processing tools for complexons are yet to be developed. 

In this work, we propose a novel \gls{CSO}. We derive its transferability properties as a limit theory of simplicial complexes. Such a theory paves the way for \gls{CSP}, making it a viable tool for analyzing signals on large and dynamic simplicial complex structures.
Our main contributions are summarized as follows:
\begin{itemize}
\item We propose the concept of a \gls{CSO} for complexons, analogous to the \gls{GRSO} for graphons.
\item We propose the raised adjacency matrix for simplicial complex and investigate its relation to the \gls{CSO} of its induced complexon.
\item We derive the transferability property of the \gls{CSO} and numerically verify it.
\item We derive the convergence of the complexon signal Fourier transform and use a mathematical model to illustrate it.
\end{itemize}

\ifextended
The rest of the paper is organized as follows. In \cref{sec:prelim}, we briefly introduce the main concepts of \gls{GSP}, graphon, and simplicial complex. In \cref{sec:complexon}, we present the concepts and properties of complexons, and define \gls{CSO} as a marginal complexon. In \cref{sec:radj}, we introduce the concept of a raised adjacency matrix for simplicial complex and relate it to \gls{CSO}. In \cref{sec:fourier}, we introduce the concept of complexon Fourier transform. In \cref{sec:convergence}, we show all the convergence results, including eigenvalue convergence of \gls{CSO} and conditioned convergence of complexon Fourier transform. Finally, in \cref{sec:exp}, we conduct two synthetic experiments to illustrate the eigenvalue convergence and Fourier transform convergence.
\else
\relax
\fi

\section{Preliminaries}
\label{sec:prelim}
In this section, we review the basic concepts of \gls{GSP}, graphons, and simplicial complexes, which are fundamental \gls{TSP} components used in setting up the theory of complexon signal processing.

\subsection{Graph And Its Shift Operators}
\label{sec:gsp}
A graph $G=(\calV(G),\calE(G))$ is a tuple, where $\calV(G)=\{v_1,v_2,\dots,v_n\}$ is the set of nodes and $\calE(G)$ is the set of edges. We define $V(G):=\abs{\calV(G)}$ and $E(G):=\abs{\calE(G)}$.
For a graph $G$, its corresponding adjacency matrix is defined as $\bA\in \set{0,1}^{n\times n}$, where $\bA_{ij}=1$ if $(v_i,v_j)\in\calE(G)$, and $0$ otherwise. For a weighted graph, $\bA_{ij}=w_{ij}$, where $w_{ij}$ is the weight of the edge $(v_i,v_j)$. In \gls{GSP}, for graph signal $(G,\bx)$, with $\bx\in\bbR^n$, a typical \gls{GSO} is the adjacency matrix $\bA$ of $G$, and the shift of the signal is $\bA \bx$. 
% Since the adjacency matrix is real and symmetric, its eigenvalues are all real numbers, and its eigenvectors form an orthonormal basis of $\bbR^n$.

% Given graphs $G_1$ and $G_2$, a homomorphism $\phi:\calV(G_1)\to\calV(G_2)$ is such that
% for any edge $(v_1,v_2)\in\calE(G_1)$, $(\phi(v_1),$\newline$\phi(v_2))\in\calE(G_2)$. Let $\hom(G_1, G_2)$ be the number of such homomorphisms. The homomorphism density is defined as
% \begin{align*}
%     t(G_1,G_2)=\frac{\hom(G_1,G_2)}{V(G_2)^{V(G_1)}}.
% \end{align*}

\subsection{Graphon}
\label{sec:gn}

The works \cite{RuiChaRib:J21, RuiCha:C20} utilize the notion of \emph{graphons} to study the transferability of \gls{GSP} among different graphs that admit similar patterns. A graphon is the limit object of a dense graph sequence \cite{Lov:12}. It is defined as a symmetric measurable function $W:[0,1]^2\to[0,1]$. 
% For a graphon, we can also define the homomorphism density. Given a graph $G$ and graphon $W$, its homomorphism density is defined as
% \begin{align*}
%     t(G,W)=\int_{[0,1]^n}\prod_{(v_i,v_j)\in\calE(G)}W(x_i,x_j)\ud x,
% \end{align*}
% where $\ud x=\prod_{i=1}^n \ud x_i$.

A graph induces a graphon via interval equipartitioning. 
\begin{Definition}
\label{def:equi}
A standard $n$-equipartition of $[0,1]$ is \newline
$\set{I_1,I_2,\dots ,I_n}$, where $I_j=[\frac{j-1}{n},\frac{j}{n})$ for $j=1,2,\dots,n-1$,
and $I_n=[\frac{n-1}{n},1]$.
\end{Definition}
Given a (weighted) graph signal $(G,\bx)$ and the (weighted) graph adjacency matrix $\bA$, the induced graphon signal $(W_G,\bX_G)$ is defined as follows. Firstly, label all vertices as $v_1,v_2,\dots,v_n$. Then let $W_G(x,y)=\bA_{ij}$ if $x\in I_i$, $y\in I_j$. Then, for $\bX_G\in L^2([0,1])$, $\bX_G=\sum_{i=1}^{n}\bx(v_i)\mathbf{1}_{I_i}$. Here $\set{I_i}_{i=1}^{n}$ is a standard $n-$equipartition, and $\bx(v_i)$ refers to the signal on $v_i$.
% It can be shown that for two graphs $F$ and $G$ we have $t(F,G)=t(F,W_G)$. 

Now we introduce the convergence of graphs. The first way to define graph convergence is through homomorphism density, which is known as \emph{left convergence}\cite{Lov:12}. 
% We say that a graph sequence $(G_n)_{n\geq1}$ is convergent if for any graph $F$, the sequence $(t(F,G_n))_{n\geq1}$ converges. Moreover from \cite{gn1}, there exists a graphon $W$ such that 
% \begin{align*}
%     \lim_{n\to\infty}t(F,G_n)=t(F,W).
% \end{align*}
% A graphon sequence $(W_n)_{n\geq1}$ is said to left converge to $W$ if 
% \begin{align*}
%     \lim_{n\to\infty}t(F,W_n)=t(F,W)
% \end{align*}
% for any graph $F$.
The second convergence definition, which is used in our research, is via cut distance, which is also called \emph{metric convergence}. Given graphons $W_1$ and $W_2$, define their labeled cut-distance as
\begin{align*}
d_\square(W_1,W_2)=\sup_{X,Y\in\calB[0,1]} \abs*{\int_{X\times Y} (W_1-W_2) \ud x\ud y},
\end{align*}
and cut distance as
\begin{align}
\delta_\square(W_1,W_2)=\inf_{\phi\in\Phi}\sup_{X,Y\in\calB[0,1]} \abs*{\int_{X\times Y} (W_1-W_2^\phi) \ud x\ud y},
\end{align}
where $\calB[0,1]$ stands for all Borel sets in $[0,1]$, $\Phi$ is the set of $[0,1]\to[0,1]$ measure-preserving transformations, and $W_2^\phi(x,y)=W_2(\phi(x),\phi(y))$.  

A graphon sequence $(W_n)_{n\geq1}$ is said to converge to $W$ in cut metric if
\begin{align*}
\lim_{n\to\infty}\delta_\square(W_n,W)=0,
\end{align*}
and the graph sequence $G_n\to W$ in cut metric if the induced graphon sequence $W_{G_n}\to W$ in cut metric.

For graphs and graphons, left convergence and metric convergence are equivalent. This can be proved using the Counting Lemma and Inverse Counting Lemma (see Theorem 2.7 and Theorem 3.7 in \cite{gn1}).

In graph spectral analysis, eigenvalue and eigenvectors of a \gls{GSO} are its fundamental components. To investigate their continuous analog for graphon, we define the graphon shift operator $T_{W}:L^2([0,1])\to L^2([0,1])$ as follows\cite{Lov:12}: 
\begin{align}
T_{W}\bX(x)=\int_0^1W(x,y)\bX(y)\ud y,
\end{align}
where $(W,\bX)$, with $\bX\in L^2([0,1])$, is a graphon signal.

It can be shown that the operator $T_W$ is linear, self-adjoint, bounded, and compact\cite{spec}. Its eigenvalues are countable, and the only possible accumulation point is 0. Its corresponding eigenvectors form an orthonormal basis in $L^2([0,1])$. By applying $T_W$ on $X$, the output on $x$ is obtained by gathering information from all other $y\in[0,1]$ with different weights.

\subsection{Simplicial Complex}
\label{sec:sc}

Given a node set $\calV=\set{v_1,v_2,\dots,v_n}$, a set $F\subseteq 2^{\calV}$ is called an $n$-node abstract simplicial complex if the following conditions hold:
\begin{itemize}
\item $(v_i)\in F$ for $i=1,\dots,n$.
\item If $\sigma\in F$ and $\sigma'\subseteq\sigma$, then $\sigma'\in F$.
\end{itemize}

A $(d+1)$-element set inside $F$ is called a $d$-dimensional simplex. The dimension of $F$, namely $\dim F$, is the highest dimension of all simplices. The $d'$-dimensional skeleton of $F$ is the subset of $F$ containing all simplices of dimension no higher than $d'$. For example, the $1$-dimensional skeleton of a simplicial complex is a graph. Let $F\tc{d}$ be the collection of all simplices with dimension $d$. 

Throughout our paper, we use $\set{v_0, v_1, \ldots,v_n}$ to denote a set with $(n+1)$ vertices, and use $\calV(\sigma)$ to represent the set of nodes in simplex $\sigma$. We use $\sigma^n=(v_0, v_1, \ldots,v_n)$ to specifically denote an $n$-dimensional simplex $\sigma^n$.

\section{Complexon with Vertex Signals}
\label{sec:complexon}

In this section, we introduce the concept of complexon and complexon shift operators.

A graphon is the limit of a sequence of graphs and can be utilized to analyze the transferability of \gls{GSP}. In order to study the transferability of \gls{TSP}, we require the graphon's counterpart for a simplicial complex, known as a \emph{complexon} \cite{LDSC}.

\begin{Definition}[Complexon]
A function 
\begin{align*}
W: \bigsqcup_{d=0}^{D}[0,1]^{d+1}\to [0,1]
\end{align*}
is called a $D$-dimensional complexon, where $D\geq1$ is an integer, if it satisfies the following properties:
\begin{enumerate}
\item It is symmetric. For $0\leq d\leq D$,
\begin{align*}
W(x_1,x_2,\dots ,x_{d+1})=W(y_1,y_2,\dots ,y_{d+1})
\end{align*} holds
if $(y_1,y_2,\dots ,y_{d+1})$ is a permutation \newline
of $(x_1,x_2,\dots ,x_{d+1})$.
\item It is measurable.
\item For the case $d=0$, $W(x)\equiv1$ for any $x\in [0,1]$.
\end{enumerate}
Furthermore, given a $D$-dimensional complexon $W$, its restriction on $[0,1]^{d+1}$ is called its $d$-dimensional
component, denoted as $W\tc{d}$.
\end{Definition}

% We can then define homomorphism densities for simplicial complexes and complexons. Specifically,
% $\phi:F\to K$ is a homomorphism if for any $\sigma\in F\tc{d}$, $\phi(\sigma)\in K\tc{d}$. Let $\hom(F,K)$ be the number of all such homomorphisms.

% \begin{Definition}[Homomorphism Density]
%     Given simplicial complexes $F$, $K$, and complexon $W$, the homomorphism densities of $F$ in $K$ and $W$ are defined as
%     \begin{align*}
%         t(F,K) &= \frac{\hom(F,K)}{V(K)^{V(F)}}, \\
%         t(F,W) &= \int_{[0,1]^{V(F)}}\prod_{\sigma\in F}W(x_\sigma)\ud x,
%     \end{align*}
%     respectively.
% \end{Definition}
% Here, $\ud x=\prod_{i=1}^{V(F)} \ud x_i$ and $x_\sigma$ stands for a tuple of variables.
% Namely, if vertices of $F$ are labeled $1,2,\ldots ,V(F)$ and $\sigma=(v_1,v_2,\dots ,v_{d+1})$
% is an ordered subset of the vertex set, then
% $W(x_\sigma)=W(x_{v_1},x_{v_2},\dots ,x_{v_{d+1}})$. This definition is well-defined as a complexon is symmetric.

Similar to a graph inducing a graphon, a simplicial complex induces a complexon.
The induced complexon $W_F$ given $D$-dimensional simplicial complex $F$ with $n$
nodes $v_1,v_2,\dots ,v_n$, is introduced in \cite{LDSC}. 
%Index $1,2,\dots ,n$ stands for a labeling of $n$ different nodes $v_1,v_2,\dots ,v_n$ of $F$.
Assume $\set{I_1,I_2,\dots ,I_n}$ is a standard $n$-equipartition of $[0,1]$.
If $x_i\in I_{g_i}$, where $g_i\in\set{1,2,\dots ,n}$, then define $W_F(x_1,x_2,\dots ,x_{d+1})=1$
if $(v_{g_1},v_{g_2},\dots ,v_{g_{d+1}})$
$\in F\tc{d}$, and 0 otherwise.

Furthermore, we can define the induced complexon signal as follows.

\begin{Definition}\label{def:cn_signal}
Given $D$-dimensional, $n-$vertex simplicial complex $F$, define its signal as pair $(F,\bs_F)$, where $\bs_F\in\bbR^n$, and the $i-$th element of $\bs_F$ is the signal on vertex $v_i$. Furthermore, define its induced complexon signal as pair $(W_F, \bX_F)$, where $W_F$ is the induced complexon of $F$, and $\bX_F=\sum_{i=1}^{n}\bs_F(v_i)\mathbf{1}_{I_i}$. Here $\set{I_i}_{i=1}^{n}$ is a standard $n-$equipartition.
\end{Definition}

% For induced complexon, the homomorphism density is retained.

% \begin{Corollary}
%     Assume $F$ and $K$ are two simplicial complexes, and $W_K$ is the induced complexon of $K$, then $t(F,K)=t(F,W_K)$.
% \end{Corollary}

Now we can define the convergence of simplicial complex sequences. Like the limits of graph sequences, we have convergence in two different senses. One is built upon homomorphism density, and the other upon the cut distance. Here we only present the definition using cut distance.

% \begin{Definition}
%     Given a simplicial complex sequence $(F_n)_{n\geq1}$ and complexon $W$, we say that $F_n\to W$ converges in homomorphism density (left convergence), if 
%     \begin{align*}
%         \lim_{n\to\infty}t(K,F_n)=t(K,W)
%     \end{align*}
%     for any simplicial complex $K$.
% \end{Definition}
\begin{Definition}
Consider a $D$-dimensional simplicial complex sequence $(F_n)_{n\geq1}$ (with their corresponding induced complexons $W_{F_n}$) and a $D$-dimensional complexon $W$. For $1\leq d\leq D$, we say that $F_n\to W$ converges in $d$-dimensional cut distance (metric convergence) if
\begin{align*}
\lim_{n\to\infty}\delta_{\square,d}(W_{F_n},W)=0,
\end{align*}
where $\delta_{\square,d}(W_{F_n},W)$ is the $d$-dimensional cut distance:
\begin{align}
\delta_{\square,d}(W_{F_n},W)=\inf_{\phi\in\Phi}\sup_{X_1,X_2,\dots,X_{d-1}\subseteq\calB[0,1]}\abs*{H_n-H^\phi},
\end{align}
with
\begin{align*}
H_n &= \int_{\Omega}W_{F_n}\tc{d}(x_1,x_2,\dots,x_{d+1})\ud x, \\
H^\phi &= \int_{\Omega}W\tc{d}(\phi(x_1),\phi(x_2),\dots,\phi(x_{d+1}))\ud x,
\end{align*}
$\ud x=\prod_{i=1}^{d+1}\ud x_i$, $\Omega=X_1\times X_2\times\dots\times X_{d+1}$.

Sometimes the \emph{labeled} $d$-dimensional cut distance is used:
\begin{align*}
d_{\square,d}(W_{F_n},W)=\sup_{X_1,X_2,\dots,X_{d-1}\subseteq\calB[0,1]}\abs*{H_n-H}.
\end{align*}
\end{Definition}
We abbreviate $W\tc{d}(\phi(x_1),\phi(x_2),\dots,\phi(x_{d+1}))$ as\newline $(W\tc{d})^\phi(x_1,x_2,\dots,x_{d+1})$.

In the context of complexons, we define $(W,\bX)$, with $\bX:[0,1]\to \bbR \in L^2([0,1])$, as a \emph{complexon signal} on complexon $W$.

The graphon shift operator $T_{W_G}$ is defined as a kernel operator. We anticipate that for a complexon component $W\tc{d}$, the complexon shift operator can be defined similarly. 
Assume $x, z_1, z_2, \ldots,z_d\in[0,1]$ represent $d+1$ vertices of a $d-$simplex. Inspired by the message-passing neural network framework \cite{Gil:C17}, we consider a $d$-dimensional complexon shift operator to aggregate information from vertices $z_1,z_2\ldots,z_d$ to $x$ to give:
\begin{align*}
\overline{\bX}(z_1,z_2,\ldots,z_d)=\sum_{i=1}^{d}\alpha_i\bX(z_i),
\end{align*}
where $\sum_{i=1}^{d}\alpha_i=1$ for normalization.

Thus, we can define the complexon shift operator as $\int_{[0,1]^d}W(x, z_1, z_2, \dots ,z_d)\overline{\bX}(z)\ud z$, where $\ud z=\prod_{i=1}^{d} \ud z_i$. After simplification of the integral, we obtain the following definition.
\begin{Definition}[Complexon Shift]\label{def:cn_shift}
Given a $D$-dimensional complexon $W$, its \gls{CSO}
at dimension $d$, denoted as $T_W\tc{d}$, is defined as
\begin{align}
T_W\tc{d}\bX(x)=\int_{0}^{1}\overline{W}\tc{d}(x,y)\bX(y)\ud y, 
\label{eq:1}
\end{align}
where 
\begin{align}
\overline{W}\tc{d}(x,y)=\int_{[0,1]^{d-1}}W(x,y,z_1,z_2,\dots ,z_{d-1})\prod_{i=1}^{d-1}\ud z_i
\end{align}
is the $d-$marginal complexon of $W$.
\end{Definition}
\section{Raised Adjacency and Complexon Shift}
\label{sec:radj}
In this section, we relate the concept of \gls{CSO} to a family of adjacency matrices, which we refer to as \emph{raised adjacency matrices}.

For \gls{GRSO}, we have the following fact.

\begin{Corollary}
\label{cor:induce_signal}
Given that $(G,\bx)$ is a graph signal, $\bA$ is the adjacency matrix, $(G,\bA\bx)$ is the shifted graph signal, $(W_G,\bX_G)$ is the induced graphon signal of $(G,\bx)$, and $(W_G,\tilde{X})$ is the induced graphon signal of $(G,\frac{1}{n}\bA\bx)$, and $T_{W_G}$ is the \gls{GRSO} of $W_G$, then, $\tilde{X}=T_{W_G}X$. 
\end{Corollary}

Given that we defined the \gls{CSO}, it is natural to look into its examples. Specifically, given complexon signal $(W_F,\bX_F)$ induced by simplicial complex signal $(F,\bs_F)$, we need to figure out the representation of \gls{CSO} $T_W\tc{d}$, or equivalently, the representation of marginal complexon $\overline{W}\tc{d}$ .

\begin{Definition}
\label{radj}
Given an $n$-node simplicial complex $K$ and dimension $d\in\set{2,\dots,\dim K}$, a $d$-raised adjacency matrix $\bN\tc{d}\in [0,1]^{n\times n}$ is such that
\begin{align}
\bN\tc{d}_{ij}=\frac{\abs*{\set*{\calV(\sigma)\in K\tc{d-2}\mid \set{v_i,v_j}\bigcup\calV(\sigma)\in K\tc{d}}}}{n^{d-1}}. \label{Ndij}
\end{align}
When $d=1$, $\bN\tc{d}_{ij}=1$ if $(v_i,v_j)\in K\tc{1}$ and equals 0 if not. In this case, $N\tc{1}$ is just the adjacency matrix considering only the 1-dimension skeleton (graph structure) of $K$.
\end{Definition}

From \cref{radj}, a raised adjacency matrix is symmetric. We also have $\bN\tc{d}_{ij}<1$ for $d\geq2$, since we are only counting simplices in the numerator and no repeated vertices are allowed. Also, if $(v_i,v_j)\notin K\tc{1}$, then we obtain 
\begin{align*}
\set*{\calV(\sigma)\in K\tc{d-2} \mid \set{v_i,v_j}\bigcup\calV(\sigma)\in K\tc{d}}=\emptyset,
\end{align*}
which implies $\bN\tc{d}_{ij}=0$.
Therefore, the raised adjacency matrix is a special weighted adjacency
matrix, whose edge weights are bounded by the entries of the standard adjacency matrix of the 1-dimensional skeleton of $F$.

With \cref{radj}, we can figure out the closed-form solution of marginal complexon induced by a simplicial complex.

\begin{Proposition}\label{prop:adj_induce_GSO}
Let $F$ be a $D$-dimensional simplicial complex with $n$ nodes $\set{v_1,v_2,\dots ,v_n}$ and $\set{I_i}_{i=1}^n$ be a standard $n$-equipartition of $[0,1]$.
Then, for any $1\leq d\leq D$, the induced $d$-dimensional marginal complexon $\overline{W}_F\tc{d}(x,y)=\bN\tc{d}_{ij}$ if $x\in I_i$, $y\in I_j$.
\end{Proposition}
\begin{proof}
See \cref{prop:adj_induce_GSO:proof}.
\end{proof}

Using \cref{prop:adj_induce_GSO}, we can calculate the specific eigenvalue and eigenfunctions of the induced marginal complexon.

\begin{Proposition}[Eigenspace of Marginal Complexon]
\label{prop:marg_cn}
Given a $D$-dimensional simplicial complex $F$ with $n$ nodes, its induced complexon is $W_F$. Let $\bN\tc{d}$ be the $d-$raised adjacency matrix of $F$, and $\set{(\lambda_i\tc{d},v_i\tc{d})\mid i\in\calL}$ be the ordered eigenvalue-eigenvector pairs, where $\calL\subseteq\bbZ\backslash\{0\}$ is a finite nonzero integer index set.
For any $1\leq d\leq D$, the $d$-dimensional \gls{CSO} is $T_{W_F}\tc{d}$. Let $\set{(\lambda_i(T_{W_F}\tc{d}), \varphi_i(T_{W_F}\tc{d})) \mid i\in\bbZ\backslash\set{0}}$ be the eigenpairs of $T_{W_F}\tc{d}$. Then, for $i\in\calL$, we have the following conclusions:
\begin{enumerate}
\item $\lambda_i(T_{W_F}\tc{d})=\frac{1}{n}\lambda_i\tc{d}$;
\item $\varphi_i(T_{W_F}\tc{d})(x)=\sqrt{n}(v_i\tc{d})_j$ if $x\in I_j$;
\item $\{\varphi_i\}_{i\in\calL}$ is an orthonormal basis of a subspace\newline $\calS\subseteq L^2([0,1])$;
\end{enumerate}
For $i\notin\calL$, we can let $\lambda_i(T_{W_F}\tc{d})=0$, $\varphi_i(T_{W_F}\tc{d})=\psi_i$, such that
$\{\psi_i\}_{i\notin\calL}$ is an orthonormal basis of $\calS^\perp$.
\end{Proposition}
\begin{proof}
See \cref{prop:marg_cn:proof}.
\end{proof}

\section{Filters And Fourier Transform}
\label{sec:fourier}

In this section, we introduce the Fourier transform for simplicial complex signals and complexon signals.

Given simplicial complex signal $(F,\bs_F)$, a shift operator $\bA$ can be used to manipulate the signal $\bs'=\bA\bs_F$. Furthermore, $\bA$ can generate a \gls{LSI} filter $\bH$:
\begin{align*}
\bs'=\bH\bs_F=\bU_{\bA}\Lambda_H\bU_{\bA}^T\bs_F
\end{align*}
where columns of $\bU_{\bA}$ are eigenvectors of $\bA$. For the case of simplicial complexes, we hereby use the raised adjacency matrix $\bN$ as the shift operator. 
% Note that $\bN$ is symmetric, it can then be diagonalized as $\bN=U\Lambda U^T$, where $U$ is an orthogonal matrix. Columns of $U$ are eigenvectors of $\bN$. 
Then the Fourier transform of simplicial complex signals can be defined as follows.
\begin{Definition}\label{def:sc_fourier}
Given $D$-dimensional simplicial complex $F$ and its corresponding signal $\bs_F$, for any $1\leq d\leq D$, define its $d-$Fourier transform as
\begin{align*}
\hat{\bs_F}\tc{d}=(\bU\tc{d})^T\bs_F,
\end{align*}
where $\bU\tc{d}$ is derived from diagonalization of $d-$raised adjacency matrix $\bN\tc{d}=\bU\tc{d}\Lambda\tc{d}(\bU\tc{d})^T$.

% The inverse $d-$Fourier transform is
% \begin{align*}
%     \bs_F=U\tc{d}\hat{\bs_F}\tc{d}.
% \end{align*}
\end{Definition}
% Given filter $H(\bN\tc{d})=\sum_{i=1}^{k}h_i(\bN\tc{d})^i$ and diagonalization $\bN\tc{d}=U\tc{d}\Lambda\tc{d}(U\tc{d})^T$, the spectral representation of filter can be derived as $\hat{H}(\bN\tc{d})=\sum_{i=1}^{k}h_i(\Lambda\tc{d})^i$. If $s'=H(\bN\tc{d})\bs_F$, then it satisfies that $\hat{s'}\tc{d}=\hat{H}(\bN\tc{d})(\bs_F)\tc{d}$.
With the definition of simplicial complex signal Fourier transform founded, the complexon version can be introduced\cite{spec}.
\begin{Definition}\label{def:cn_fourier}
Given $D$-dimensional complexon signal $(W,\bX)$, for any $1\leq d\leq D$, assume that \gls{CSO} is $T_W\tc{d}$, define $d-$Fourier transform of $\bX$ as $\hat{\bX}\tc{d}$:
\begin{align*}
[\hat{\bX}\tc{d}]_m=\int_0^1\bX(t)\varphi_m\tc{d}(t)\ud t,
\end{align*}
where $m$ is a non-zero integer index and $\varphi_m\tc{d}$ is the eigenfunction of the eigendecomposition of $T_W\tc{d}$: $\set{(\lambda_i\tc{d},\varphi_i\tc{d})}_{i\in\bbZ\backslash\set{0}}$.
% The inverse $d-$Fourier transform is
% \begin{align*}
%     f=\sum_{i\in\bbZ\backslash\set{0}}[\hat{f}\tc{d}]_i\varphi_i\tc{d}
% \end{align*}
\end{Definition}

\section{Convergence}\label{sec:convergence}

In this section, we study the convergence of complexons. 

Firstly, we prove that if a sequence of simplicial complexes converges to a complexon, then the eigenvalues of their induced \gls{CSO}s also converge. 

% However, we first need to prove that the graphon-based cut distance between marginal complexons can be controlled by the complexon cut distance of original complexons. 

% Combining \cref{lma:cut_dist_control} with \cref{prop:adj_induce_GSO}, we immediately yield the following corollary, which further leads to the eigenvalue convergence of simplicial complex sequence to the limit complexon.

% \begin{Corollary}\label{cor:cn_to_gn}
%     If $D$-dimensional simplicial complex sequence $F_k\to W$ under the cut distance of any dimension, then for any dimension $1\leq d\leq D$, $G_k\tc{d}\to\overline{W}\tc{d}$ converges in graphon cut distance if $G_k\tc{d}$ is a weighted graph with adjacency matrix $\bN_k\tc{d}$. 
% \end{Corollary}

% \begin{Corollary}\label{cor:conv_in_op_norm}
% Given $1\leq d\leq D$, suppose the $D$-dimensional complexon sequence
% $W_k\to W$ under the cut distance of any dimension $d$, then $T_{W_k}\tc{d}\to T_W\tc{d}$ in operator norm.
% \end{Corollary}
% \begin{proof}
%     According to \cref{lma:cut_dist_control}, convergence of $W_k\to W$ in complexon cut distance of any dimension implies $\overline{W}_k\tc{d}\to\overline{W}\tc{d}$ in graphon cut distance for any dimension $d$. Theorem holds referring to \cite[Proposition A]{reply}.
% \end{proof}
\begin{Theorem}\label{thm:adj_induce_GSO}
Given $1\leq d\leq D$, suppose the $D$-dimensional simplicial complex sequence
$F_k\to W$ under the cut distance of any dimension.
For each $F_k$, suppose the eigenvalues of $T_{W_{F_k}}\tc{d}$ are $\set{\lambda_{i}\tc{d,k}\mid i\in\bbZ\backslash\set{0}}$ and the eigenvalues of $T_W\tc{d}$ are $\set{\lambda_{i}\tc{d}\mid i\in\bbZ\backslash\set{0}}$.
Then, for any $i\in\bbZ\backslash\{0\}$,
\begin{align*}
\lim_{k\to\infty}\lambda_{i}\tc{d,k}=\lambda_{i}\tc{d}.
\end{align*}
\end{Theorem}
\begin{proof}
See \cref{thm:adj_induce_GSO:proof}.
\end{proof}

\cref{thm:adj_induce_GSO} implies transferability of simplicial complex signal processing. However, to further discuss the convergence of the Fourier transform, more limitations are needed.
\begin{Definition}\label{def:signal_conv}
Define simplicial complex signal sequence $\set{(F_n,\bX_{F_n})}_{n=1}^{\infty}$ converges to complexon signal $(W,\bX)$ if there exists a sequence of \emph{admissible node permutation} (adapted from \cite[Definition 1]{RuiChaRib:J21}) $\{\pi_n\}$ such that $\ud_{\square,d}(W_{\pi_n(F_n)},W)\to0$ for any dimension $d$ and $\Vert \pi_n(\bX_n)-\bX \Vert_2\to 0$. Here $\bX_n$ is the induced complexon signal of $\bX_{F_n}$.
\end{Definition}
\begin{Definition}\label{def:non_derogatory}
$D$-dimensional complexon $W$ is fully non-derogatory if all non-zero eigenvalues of $T_W\tc{d}$ have multiplicity one for any dimension $d$.
\end{Definition}
\begin{Definition}\label{def:bandlimited}
$D$-dimensional complexon signal $(W,\bX)$ is $c-$bandlimited at dimension $d$ if $[\hat{\bX}\tc{d}]_m=0$ for any $\abs{\lambda_m\tc{d}}<c$. Furthermore, it is $c-$bandlimited uniformly if $[\hat{\bX}\tc{d}]_m=0$ for any $\abs{\lambda_m\tc{d}}<c$ at any dimension $d$.
\end{Definition}

\begin{Theorem}[Convergence of Fourier transform]\label{thm:conv_cft}
Given that a sequence of $D$-dimensional simplicial complex signal $(F_n,\bs_{F_n})$ converges to complexon signal $(W,\bX)$ under admissible permutation $\{\pi_n\}$ where
\begin{itemize}
\item $W$ is fully non-derogatory;
\item $\bX$ is $c-$bandlimited uniformly.
\end{itemize}
Denote $\bX_n$ as the induced complexon signal of $\bs_{F_n}$. Then at any dimension $d$, $\hat{\pi_n(\bX_n)}\tc{d}\to\hat{\bX}\tc{d}$ pointwise.
\end{Theorem}
\begin{proof}
See \cref{thm:conv_cft:proof}.
\end{proof}

\section{Experiments}
\label{sec:exp}

\subsection{Synthetic Experiment for Eigenvalue Convergence}

To corroborate \cref{thm:adj_induce_GSO}, we generate a synthetic example of a 2-dimensional complexon $W: \bigsqcup_{d=0}^{2}[0,1]^{d+1}\to [0,1]$:
\begin{align*}
W\tc{0}(x)&\equiv1, \\
W\tc{1}(x,y)&\equiv1, \\
W\tc{2}(x,y,z)&=\frac{x+y+z}{3}.
\end{align*}
Given node number $n$ (in our experimental setting, $n=6,7,\ldots,149$), a sampled simplicial complex $F_n$ is constructed as follows. First, draw $n$ sample points $\set{x_1,x_2,\cdots,x_n}$ \gls{iid} from the uniform distribution $\dist{Unif}[0,1]$. Then, create its node set $F_n\tc{0}=\set{v_1,v_2,\cdots,v_n}$ and edge set $F_n\tc{1}=\set{(v_i,v_j)\mid i,j\in\set{1,2,\cdots,n},i\neq j}$. For each 2-dimensional simplex $(v_i,v_j,v_k)$, the probability such that $(v_i,v_j,v_k)\in F_n\tc{2}$ is $W\tc{2}(x_i,x_j,x_k)$. According to \cite{LDSC}, $F_n\to W$ under the cut distance of any dimension.

We investigate the eigenvalue convergence behavior of \gls{CSO} $T_{W_{F_n}}\tc{2}$. For $W\tc{2}$, its marginal complexon is $\overline{W}\tc{2}(x,y)=\frac{x}{3}+\frac{y}{3}+\frac{1}{6}$ according to \cref{def:cn_shift}.
The only two non-zero eigenvalues for \gls{CSO} $T_W\tc{2}$ are $\lambda_1=\frac{1}{4}+\frac{\sqrt{93}}{36}$, $\lambda_{-1}=\frac{1}{4}-\frac{\sqrt{93}}{36}$. So we anticipate that for sequence $T_{W_{F_n}}\tc{2}$, $\lambda_{1}\tc{2,n}\to\lambda_1$, $\lambda_{-1}\tc{2,n}\to\lambda_{-1}$, and $\lambda_{l}\tc{2,n}\to0$ for any $l\in\bbZ\backslash\set{0,1,-1}$. \cref{fig:eig} shows the convergence of the eigenvalues for $\lambda_{i}\tc{2,n}$ for $i=1,2,-1,-2$. This experiment result verifies \cref{thm:adj_induce_GSO} and indicates the transferability of simplicial complex sequences converging to a complexon.

\begin{figure}[tb]
\centering
\includegraphics[width=.85\columnwidth,trim=2cm 0 2cm 0, clip]{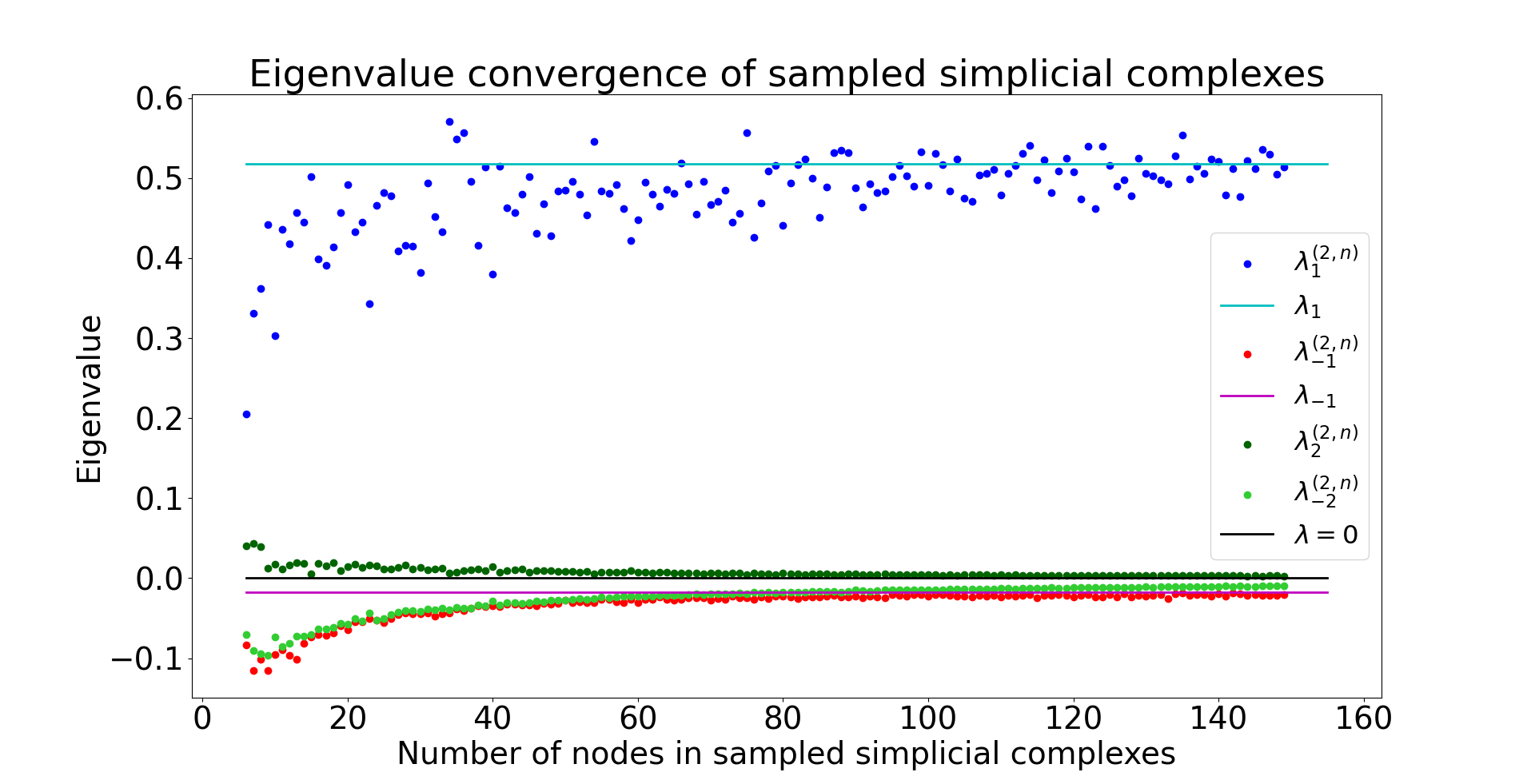}
\caption{Eigenvalue convergence of $\lambda_i\tc{2,n}$ for $i=1,2,-1,-2$.}
\label{fig:eig}
\includegraphics[width=.85\columnwidth,trim=.5cm 0 .5cm 0, clip]{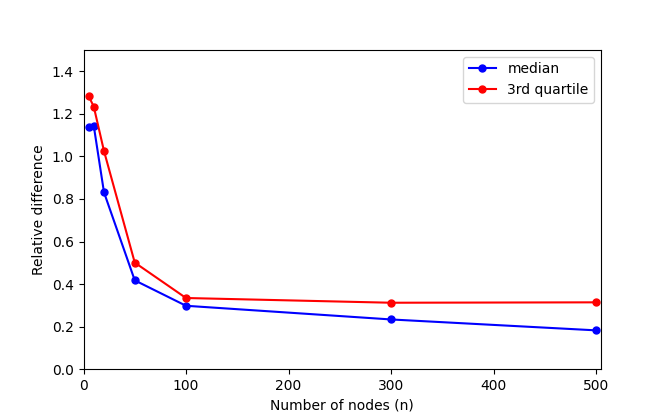}
\caption{Convergence of relative difference of Fourier transform of signals for $n=5,10,20,50,100,300,500$. Each $n$ is repeated with 50 trials.}
\label{fig:fourier}
\end{figure}

\subsection{Experiment for Fourier Transform Convergence}

To corroborate \cref{thm:conv_cft}, we adopt a synthetic model, inspired by \cite[Experiment S2]{RuiChaRib:J21}, to illustrate the convergence of the Fourier transform. 

Consider a field produced by a single source in Euclidean space $\bbR^2$. It can model a sound intensity field produced by a noise source or environmental defection produced by a pollution source. To characterize, an $n-$node simplicial complex sensor network $F$ is placed in the field with vertices $\set{v_i}_{i=1}^{n}$ and coordinate $\set{(x_i,y_i)}_{i=1}^{n}\sim\dist{Unif}[-1,1]^2$, \gls{iid}. Define normalized vertical distance $d_i=\frac{y_i+1}{2}$ to characterize the relationship between sensors. We then create $2$-dimensional simplicial complex $F$ with $F\tc{0}=\set{v_1,v_2,...,v_{n}}$, $F\tc{1}$ fully-connected, and
\begin{align*}
p(\set{v_i,v_j,v_k}\in F\tc{2})=\mathrm{exp}(-\beta\overline{d}),
\end{align*}
where $\beta>0$ is a constant and $\overline{d}=\max\set{\abs{d_i-d_j},\abs{d_j-d_k},\abs{d_k-d_i}}$.

At each sensor node $v_i$ with coordinate $(x_i,y_i)$, we assume that its captured signal follows normal distribution $[\bs_{F}]_i\sim\calN(y_i,\sigma^2)$, conforming the diffusion of source. Induced signal satisfies $\bX_{F}\in L^2([0,1])$ and its 2-Fourier transform is $\hat{\bX_{F}}\tc{d}$. We choose $\beta=0.5$ and $\sigma=0.2$ in our setting.

To illustrate Fourier transform convergence, we sample another series of simplicial complexes $E$ in the same way, such that $F$ and $E$ have the same number of nodes $n$. The coordinates, simplices, and signals are re-sampled. After getting $\bX_{E}$ and Fourier transform $\hat{\bX_{E}}\tc{d}$, we compare $\hat{\bX_{E}}\tc{d}$ and $\hat{\bX_{F}}\tc{d}$, and then find a node permutation $\pi$ such that $\Vert\hat{X_{E}}\tc{d}-\hat{X}_{\pi(F)}\tc{d}\Vert_2$ is minimized. The aforementioned process is repeated 50 times and the relative difference $\frac{\Vert\hat{\bX_{E}}\tc{d}-\hat{\bX}_{\pi(F)}\tc{d}\Vert_2}{\Vert\hat{\bX_{E}}\tc{d}\Vert_2}$ is calculated. The median and 3rd quartile of these 50 distances are acquired. When $n$ increases, if the median and 3rd quartile converge to zero, the convergence of the Fourier transform is then illustrated. We run tests for $n=5,10,20,50,100,300,500$ and obtain the convergence tendency of the relative difference of Fourier transform of signals in \cref{fig:fourier}. This experiment demonstrates the convergence behavior of the Fourier transform empirically, which aligns with \cref{thm:conv_cft}.

\section{Conclusion}
\label{sec:conclusion}

In this work, we proposed a type of complexon shift operator based on marginal complexons and found raised adjacency matrix as its corresponding shift operator for simplicial complexes. We proved the eigenvalue convergence of \gls{CSO} and the convergence of the Fourier transform of complexon signals. These two conclusions are further supported by numerical experiments on sampled simplicial complex sequences, with models close to real-life applications. The convergence of \gls{CSO}, its eigenvalue, and Fourier transform of complexon signals implies the transferability of \gls{SCSP} on vertex signals, which suggests the potential application of complexon signal processing on large or dynamic simplicial complex networks.

\section{Acknowledgements}
This research is supported by the Singapore Ministry of Education Academic Research Fund Tier 2 grant MOE-T2EP20220-0002.

\appendices

\section{Proof of \cref{prop:adj_induce_GSO}}\label[Appendix]{prop:adj_induce_GSO:proof}

For $d=1$, $\overline{W}_F\tc{1}=W_F\tc{1}$, which is identical to the graphon of the 1-dimensional skeleton of $F$. The matrix $\bN\tc{1}$ is the adjacency matrix of the 1-dimensional skeleton of $F$. So the proposition holds by relating the graph and its induced graphon.

For $d\geq 2$, first, we consider the $d$-dimensional component of complexon $W$. That is, $W\tc{d}(x,y,z_1,z_2,\dots ,z_{d-1})$. According to \cref{def:cn_shift}, 
\begin{align*}
\overline{W}_F\tc{d}(x,y)=\int_{[0,1]^{d-1}}W_F(x,y,z_1,z_2,\dots ,z_{d-1})\prod_{i=1}^{d-1}\ud z_i.
\end{align*}

To calculate the integral, we should first split the integral intervals:
\begin{align*}
\int_{[0,1]^{d-1}}\prod_{i=1}^{d-1}\ud z_i=\sum_{k_1,k_2,\dots ,k_{d-1}=1}^n\int_{I_{k_1}\times I_{k_2}\times\dots \times I_{k_{d-1}}}\prod_{i=1}^{d-1}\ud z_i.
\end{align*}
Assume $x\in I_i$, $y\in I_j$. We are going to prove $\overline{W}_F\tc{d}(x,y)=N\tc{d}_{ij}$. 
According to the definition, we have 
\begin{align*}
W(x,y,z_1,z_2,\dots ,z_{d-1})=1
\end{align*}
if
\begin{align*}
(v_i,v_j,v_{k_1},v_{k_2},\dots ,v_{k_{d-1}})\in K\tc{d}
\end{align*}
and 0 otherwise. Since all $k-$entries range from $1$ to $n$, we are counting all $d$-dimensional simplices in $K\tc{d}$ containing vertices $v_i$ 
and $v_j$. And for the hyper-volume of each integral interval, it should be $(\frac{1}{n})^{d-1}$. 
So by \cref{radj}, we have
\begin{align*}
\overline{W}_F\tc{d}(x,y) &= \frac{\abs*{\set{\calV(\sigma)\in K\tc{d-2}\mid\set{v_i,v_j}\bigcup\calV(\sigma)\in K\tc{d}}}}{n^{d-1}}\\
&= \bN\tc{d}_{ij}.\\
\end{align*}

\section{Proof of \cref{prop:marg_cn}}\label[Appendix]{prop:marg_cn:proof}
Note that since a marginal complexon is a graphon, the $d$-dimensional \gls{CSO} has the same properties as
the graphon shift operator: it is linear, self-adjoint, and compact\cite{RuiChaRib:J21}. 

To prove the first two conclusions, we only need to verify that $T_{W_F}\tc{d}\varphi_i=\frac{1}{n}\lambda_i\tc{d}\varphi_i$ holds for any $i\in\calL$. To do this, we set up standard $n$-equipartition $\set{I_1,I_2,\dots ,I_n}$. In this case, for any $x\in I_j$, $j\in\set{1,2,\dots ,n}$, we have
\begin{align*}
T_{W_F}\tc{d}\varphi_i(x) &= \int_0^1\overline{W}_F\tc{d}(x,y)\varphi_i(y)\ud y\\
&= \sum_{k=1}^{n}\int_{I_k}\overline{W}_F\tc{d}(x,y)\varphi_i(y)\ud y\\
&= \sqrt{n}\sum_{k=1}^{n}\int_{I_k}N_{jk}\tc{d}(v_i\tc{d})_k\ud y\\
&= \sqrt{n}\sum_{k=1}^{n}(\int_{I_k}\ud y)N_{jk}\tc{d}(v_i\tc{d})_k\\
&= \sqrt{n}\cdot\frac{1}{n}(N\tc{d}v_i\tc{d})_j\\
&= \frac{1}{n}(\lambda_i\tc{d}\sqrt{n}v_i\tc{d})_j\\
&= \frac{1}{n}\lambda_i\tc{d}\varphi_i(x).
\end{align*}

For the third conclusion, we need to prove $\langle\varphi_i,\varphi_j\rangle=\delta_{ij}$,
for any $i,j\in\calL$, where $\delta$ is the Kronecker delta. Given that $N\tc{d}$ is a real
symmetric matrix, $\langle v_i,v_j\rangle=\sum_{k=1}^{n}(v_i)_k(v_j)_k=\delta_{ij}$, we have
\begin{align*}
\langle\varphi_i,\varphi_j\rangle &= \int_0^1\varphi_i(x)\varphi_j(x)\ud x \\
&= \sum_{k=1}^n\int_{I_k}\varphi_i(x)\varphi_j(x)\ud x \\
&= \sum_{k=1}^n(\int_{I_k}\ud x)\sqrt{n}(v_i)_k\sqrt{n}(v_j)_k \\
&= \sum_{k=1}^{n}(v_i)_k(v_j)_k \\
&= \delta_{ij}, \\
\end{align*}
which concludes the proof.

\section{Proof of \cref{thm:adj_induce_GSO}}\label[Appendix]{thm:adj_induce_GSO:proof}

We start with a preliminary lemma. 

\begin{Lemma}\label{lma:cut_dist_control}
Given two $D$-dimensional complexon $W_1$, $W_2$, with their $d$-dimensional components $W_1\tc{d}$, $W_2\tc{d}$, $1\leq d\leq D$, it holds that
\begin{align*}
d_{\square}(\overline{W_1}\tc{d},\overline{W_2}\tc{d})&\leq d_{\square,d}(W_1,W_2)\\
\delta_{\square}(\overline{W}_{1}\tc{d},\overline{W}_2\tc{d})&\leq\delta_{\square,d}(W_1,W_2).
\end{align*}
\end{Lemma}
\begin{proof}
For the first inequality,

\begin{align*}
d_{\square}(\overline{W_1}\tc{d},\overline{W_2}\tc{d}) = \sup_{X,Y\subseteq\calB[0,1]}\abs*{J_1-J_2},
\end{align*}
where
\begin{align*}
J_1&\triangleq\int_{X\times Y}\overline{W_1}\tc{d}(x,y)\ud x\ud y,\\
J_2&\triangleq\int_{X\times Y}\overline{W_2}\tc{d}(x,y)\ud x\ud y.
\end{align*}
By definition of marginal complexon, we have
\begin{align*}
J_1 &=\int_{\Omega_1}W_1\tc{d}(x,y,z_1,\dots,z_{d-1})\ud x\ud y\ud z;\\
J_2 &=\int_{\Omega_1}W_2\tc{d}(x,y,z_1,\dots,z_{d-1})\ud x\ud y\ud z.
\end{align*}
where $\ud z=\prod_{i=1}^{d-1}\ud z_i$, $\Omega_1=X\times Y\times[0,1]^{d-1}$.

Denote
\begin{align*}
d_{\square,d}(W_1,W_2)=\sup_{X,Y,Z_1,\dots,Z_{d-1}\subseteq\calB[0,1]}\abs*{J_1'-J_2'},
\end{align*}
where 
\begin{align*}
J_1' &= \int_{\Omega_2}W_1\tc{d}(x,y,z_1,\dots,z_{d-1})\ud x\ud y\ud z ;\\
J_2' &= \int_{\Omega_2}W_2\tc{d}(x,y,z_1,\dots,z_{d-1})\ud x\ud y\ud z,
\end{align*}
and $\ud z=\prod_{i=1}^{d-1}\ud z_i$, $\Omega_2=X\times Y\times Z_1\dots \times Z_{d-1}$.

Note that $\Omega_1$ is a special condition of $\Omega_2$ if we let $Z_1, Z_2,\dots,Z_{d-1}=[0,1]$. So we have
\begin{align*}
\sup_{X,Y\subseteq\calB[0,1]}\abs*{J_1-J_2}\leq\sup_{X,Y,Z_1,\dots ,Z_{d-1}\subseteq\calB[0,1]}\abs*{J_1'-J_2'}.
\end{align*}
That is, 
\begin{align*}
d_{\square}(\overline{W_1}\tc{d},\overline{W_2}\tc{d})\leq d_{\square,d}(W_1,W_2).
\end{align*}

For the second inequality, assume $\phi\in\Phi$,
\begin{align*}
\delta_{\square}(\overline{W}_1\tc{d},\overline{W}_2\tc{d}) = \inf_\phi\sup_{X,Y\subseteq\calB[0,1]}\abs*{H_1-H_2},
\end{align*}
where
\begin{align*}
H_1&\triangleq \int_{X\times Y}\overline{W}_1\tc{d}(x,y)\ud x\ud y,\\
H_2&\triangleq\int_{X\times Y}\overline{W}_2\tc{d}(\phi(x),\phi(y))\ud x\ud y.
\end{align*}
By definition of marginal complexon, we have
\begin{align*}
H_1 &=\int_{\Omega_1}W_{1}\tc{d}(x,y,z_1,\dots,z_{d-1})\ud x\ud y\ud z;\\
H_2 &=\int_{\Omega_1}W_2\tc{d}(\phi(x),\phi(y),z_1,\dots,z_{d-1})\ud x\ud y\ud z.
\end{align*}
where $\ud z=\prod_{i=1}^{d-1}\ud z_i$, $\Omega_1=X\times Y\times[0,1]^{d-1}$.

For any measurable functions $f:[0,1]\to\bbR$, we have $\int_0^1f(\phi(x))\ud x=\int_0^1f(x)\ud x$\cite{measure}. Therefore we have
\begin{align*}
H_2=\int_{X\times Y\times[0,1]^{d-1}}(W_2\tc{d})^\phi(x,y,z_1,\dots,z_{d-1})\ud x\ud y\ud z
\end{align*}
by substituting all $z_i$ with $\phi(z_i)$ without changing the value of the integral.

Compare the resulting term with the corresponding part of $\delta_{\square,d}(W_1, W_2)$. Denote
\begin{align*}
\delta_{\square,d}(W_1,W_2)=\inf_{\phi}\sup_{X,Y,Z_1,\dots ,Z_{d-1}\subseteq\calB[0,1]}\abs*{H_1'-H_2'},
\end{align*}
where 
\begin{align*}
H_1' &= \int_{\Omega_2}W_1\tc{d}(x,y,z_1,\dots,z_{d-1})\ud x\ud y\ud z ;\\
H_2' &= \int_{\Omega_2}(W_2\tc{d})^\phi(x,y,z_1,\dots,z_{d-1})\ud x\ud y\ud z,
\end{align*}
and $\ud z=\prod_{i=1}^{d-1}\ud z_i$, $\Omega_2=X\times Y\times Z_1\dots \times Z_{d-1}$.

Note that $\Omega_1$ is a special condition of $\Omega_2$ if we let $Z_1, Z_2,\dots,Z_{d-1}=[0,1]$. So for any $\phi\in\Phi$, we have
\begin{align*}
\sup_{X,Y\subseteq\calB[0,1]}\abs*{H_1-H_2}\leq\sup_{X,Y,Z_1,\dots ,Z_{d-1}\subseteq\calB[0,1]}\abs*{H_1'-H_2'}.
\end{align*}
Taking the infimum of $\phi$ and then we get 
\begin{align*}
\delta_{\square}(\overline{W}_1\tc{d},\overline{W}_2\tc{d})\leq\delta_{\square,d}(W_1,W_2).
\end{align*}
\end{proof}

We next prove \cref{thm:adj_induce_GSO}. According to \cref{lma:cut_dist_control},
\begin{align*}
\delta_{\square}(\overline{W}_{F_k}\tc{d},\overline{W}\tc{d})\leq\delta_{\square,d}(W_{F_k},W).
\end{align*}
So the convergence of $F_k\to W$ in complexon cut distance of any dimension implies $\overline{W}_{F_k}\tc{d}\to\overline{W}\tc{d}$ in graphon cut distance. Then by directly applying \cite[Theorem 6.7]{gn2} we obtain the desired result.

\section{Proof of \cref{thm:conv_cft}}\label[Appendix]{thm:conv_cft:proof}
Before the proof, we hereby state the following corollary, which is a direct result from \cref{lma:cut_dist_control}.
\begin{Corollary}\label{cor:cn_to_gn}
If $D$-dimensional simplicial complex sequence $F_k\to W$ under the cut distance of any dimension, then for any dimension $1\leq d\leq D$, $G_k\tc{d}\to\overline{W}\tc{d}$ converges in graphon cut distance if $G_k\tc{d}$ is a weighted graph with adjacency matrix $\bN_k\tc{d}$. 
\end{Corollary}

With the help of \cref{cor:cn_to_gn}, we can link properties of marginal complexons to graphons induced by weighted graphs.
We now prove \cref{thm:conv_cft}.
Denote $W_{\pi_n(F_n)}$ as $W_n$, $\pi_n(\bX_n)$ as $X_n$, and $\bX$ as $X$. Then signal convergence is equivalent to $d_{\square,d}(W_n,W)\to0$ and $\Vert X_n-X\Vert_2\to 0$. We now need to prove $\hat{X_n}\tc{d}\to\hat{X}\tc{d}$ pointwise.

Since $d_{\square,d}(W_n,W)\to0$ is given by the signal convergence, using \cref{lma:cut_dist_control}, $d_{\square}(\overline{W_n}\tc{d},\overline{W}\tc{d})\to0$ holds. 

Assume that $d$-dimensional component of $\pi_n(F_n)$ induces $d-$raised adjacency matrix $\bN\tc{d}$, then this adjacency matrix can induce a weighted graph $G_n\tc{d}$, which can further induce graphon $W_{G_n\tc{d}}$. Then $\overline{W_n}\tc{d}=W_{G_n\tc{d}}$ according to the definition of marginal complexon. Moreover, according to \cref{def:cn_shift}, $T_{W_n}\tc{d}$ is induced by $\overline{W_n}\tc{d}$, and we can induce graphon shift operator $T_{G_n\tc{d}}$ by $W_{G_n\tc{d}}$. Two operators are induced by the same function in the same way, so they are identical. 

Denote eigenpairs of $T_{W_n}\tc{d}$ (or equivalently, $T_{G_n\tc{d}}$) as $\set{(\lambda_{i}\tc{d,n},\varphi_i\tc{d,n})}_{i\in\bbZ\backslash\set{0}}$ and eigenpairs of $T_{W}\tc{d}$ as $\set{(\lambda_{i}\tc{d},\varphi_i\tc{d})}_{i\in\bbZ\backslash\set{0}}$. Thus, we can adapt \cite[Lemma 3]{RuiChaRib:J21}: define $\calC=\set{i\in\bbZ\backslash\set{0}\mid\abs*{\lambda_i\tc{d}}\geq c}$.

Case I. For $i\in\calC$, for any $\varepsilon>0$, there exists $N\in\bbN_{+}$ such that for any $n>N$, 
\begin{align*}
\Vert\varphi_i\tc{d,n}-\varphi_i\tc{d}\Vert_2<\frac{\varepsilon}{2\Vert X\Vert_2},\\
\Vert X_n-X\Vert_2<\frac{\varepsilon}{2}.
\end{align*}
Now denote the inner product in $L^2[0,1]$ Hilbert space as $\langle f,g\rangle=\int_0^1f(t)g(t)\ud t$. Then according to \cref{def:cn_fourier}, for any $i\in\calC$,
\begin{align*}
&\quad\abs{[\hat{X_n}\tc{d}]_i-[\hat{X}\tc{d}]_i}\\
&=\abs{\langle X_n,\varphi_i\tc{d,n}\rangle-\langle X,\varphi_i\tc{d}\rangle}\\
&=\abs{\langle X_n-X,\varphi_i\tc{d,n}\rangle-\langle X,\varphi_i\tc{d,n}-\varphi_i\tc{d}\rangle} \\
&\leq \Vert X_n-X\Vert_2\Vert\varphi_i\tc{d,n}\Vert_2+\Vert X\Vert_2\Vert\varphi_i\tc{d,n}-\varphi_i\tc{d}\Vert_2 \\
&<\frac{\varepsilon}{2}\cdot1+\Vert X\Vert_2\cdot\frac{\varepsilon}{2\Vert X\Vert_2} \\
&=\varepsilon.
\end{align*}
Here, $\Vert\varphi_i\tc{d,n}\Vert_2=1$ because it is an eigenfunction in an orthonormal basis.

Case II. For $i\notin\calC$, note that $f$ (or, $X$) is $c-$bandlimited uniformly (so specifically $c-$bandlimited at dimension $d$), $[\hat{X}\tc{d}]_i=\langle X,\varphi_i\tc{d}\rangle=0$. Since this holds for any $i\notin\calC$, $X$ is then perpendicular to the subspace $\calS=\mathrm{span}\set{\varphi_m\tc{d}}_{m\notin\calC}$. Since \cite[Lemma 3]{RuiChaRib:J21} shows that $\varphi_i\tc{d,n}\to\psi_i\tc{d}\in\calS$ weakly, we have $\langle X,\psi_i\tc{d}\rangle=0$. Then, for any $\varepsilon>0$, there exists $N\in\bbN_{+}$ such that for any $n>N$,
\begin{align*}
\abs{\langle\varphi_i\tc{d,n}-\psi_i\tc{d},X\rangle}<\frac{\varepsilon}{2}\\
\Vert X_n-X\Vert_2<\frac{\varepsilon}{2}
\end{align*}
Then, note that $\langle X,\varphi_i\tc{d}\rangle=\langle X,\psi_i\tc{d}\rangle=0$,
\begin{align*}
&\quad\abs{[\hat{X_n}\tc{d}]_i-[\hat{X}\tc{d}]_i}\\
&=\abs{\langle X_n,\varphi_i\tc{d,n}\rangle-\langle X,\varphi_i\tc{d}\rangle}\\
&=\abs{\langle X_n-X,\varphi_i\tc{d,n}\rangle+\langle X,\varphi_i\tc{d,n}-\psi_i\tc{d}\rangle}\\
&\leq \Vert X_n-X\Vert_2\cdot\Vert\varphi_i\tc{d,n}\Vert_2+\abs{\langle X,\varphi_i\tc{d,n}-\psi_i\tc{d}\rangle}\\
&<\frac{\varepsilon}{2}\cdot1+\frac{\varepsilon}{2}\\
&=\varepsilon.
\end{align*}
To sum up, for any $i\notin\bbZ\backslash\set{0}$, $\abs{[\hat{X_n}\tc{d}]_i-[\hat{X}\tc{d}]_i}\to 0$. That is, $\hat{X_n}\tc{d}\to\hat{X}\tc{d}$ pointwise.

\begin{Remark}
Given $D$-dimensional simplicial complex $F$, its corresponding signal $\bs_F$, for any $1\leq d\leq D$, and its $d-$Fourier transform $\hat{\bs_F}\tc{d}=(\bU\tc{d})^T\bs_F$, we can also define the inverse $d-$Fourier transform as $\bs_F=\bU\tc{d}\hat{\bs_F}\tc{d}$. Similarly, given $D$-dimensional complexon signal $(W,\bX)$, for any $1\leq d\leq D$, assume that \gls{CSO} is $T_W\tc{d}$, and the $d-$Fourier transform of $\bX$ as $\hat{\bX}\tc{d}$ is $[\hat{\bX}\tc{d}]_m=\int_0^1\bX(t)\varphi_m\tc{d}(t)\ud t$, then the inverse $d-$Fourier transform can be defined as $\bX=\sum_{i\in\bbZ\backslash\set{0}}[\hat{\bX}\tc{d}]_i\varphi_i\tc{d}$.

Also, given filter $\bH(\bN\tc{d})=\sum_{i=1}^{k}h_i(\bN\tc{d})^i$ and diagonalization $\bN\tc{d}=\bU\tc{d}\Lambda\tc{d}(\bU\tc{d})^T$, the spectral representation of filter can be derived as $\hat{\bH}(\bN\tc{d})=\sum_{i=1}^{k}h_i(\Lambda\tc{d})^i$. If $\bs'=\bH(\bN\tc{d})\bs_F$, then it satisfies that $\hat{\bs'}\tc{d}=\hat{\bH}(\bN\tc{d})(\bs_F)\tc{d}$.
\end{Remark}

\bibliographystyle{IEEEbib}
\bibliography{IEEEabrv,StringDefinitions,strings}

\end{document}

%% file: preamble.tex
\usepackage[T1]{fontenc}
\usepackage[utf8]{inputenc}
\usepackage{mathtools}

\usepackage{amssymb,mathrsfs}% Typical maths resource packages
\usepackage{amsthm}
\usepackage{bm}
\usepackage{scalerel}
\usepackage{nicefrac}
\usepackage{microtype} 
\usepackage[shortlabels]{enumitem}
\usepackage{graphicx}
\usepackage{epstopdf}
\DeclareGraphicsExtensions{.eps,.png,.jpg,.pdf}

\usepackage{url}
\usepackage{colortbl}
\usepackage{booktabs}
\usepackage{multirow}
\usepackage{colortbl,xcolor}
\usepackage{soul}
\usepackage{xparse,xstring}
\usepackage{calc}
\usepackage{etoolbox}

\makeatletter
\@ifpackageloaded{natbib}{
	\relax
}{
	\usepackage{cite}
}
\makeatother

\usepackage{array}
\newcolumntype{L}[1]{>{\raggedright\let\newline\\\arraybackslash\hspace{0pt}}m{#1}}
\newcolumntype{C}[1]{>{\centering\let\newline\\\arraybackslash\hspace{0pt}}m{#1}}
\newcolumntype{R}[1]{>{\raggedleft\let\newline\\\arraybackslash\hspace{0pt}}m{#1}}

\makeatletter
\let\MYcaption\@makecaption
\makeatother
\usepackage[font=footnotesize]{subcaption}
\makeatletter
\let\@makecaption\MYcaption
\makeatother

\usepackage{glossaries}
\makeatletter
\sfcode`\.1006

\let\oldgls\gls
\let\oldglspl\glspl

\newcommand\fussy@ifnextchar[3]{%
	\let\reserved@d=#1%
	\def\reserved@a{#2}%
	\def\reserved@b{#3}%
	\futurelet\@let@token\fussy@ifnch}
\def\fussy@ifnch{%
	\ifx\@let@token\reserved@d
		\let\reserved@c\reserved@a
	\else
		\let\reserved@c\reserved@b
	\fi
	\reserved@c}

\renewcommand{\gls}[1]{%
\oldgls{#1}\fussy@ifnextchar.{\@checkperiod}{\@}}
\renewcommand{\glspl}[1]{%
\oldglspl{#1}\fussy@ifnextchar.{\@checkperiod}{\@}}

\newcommand{\@checkperiod}[1]{%
	\ifnum\sfcode`\.=\spacefactor\else#1\fi
}

\robustify{\gls}
\robustify{\glspl}
\makeatother

\newacronym{wrt}{w.r.t.}{with respect to}
\newacronym{RHS}{R.H.S.}{right-hand side}
\newacronym{LHS}{L.H.S.}{left-hand side}
\newacronym{iid}{i.i.d.}{independent and identically distributed}
\newacronym{SOTA}{SOTA}{state-of-the-art}
\usepackage{float}

\ifx\notloadhyperref\undefined
	\ifx\loadbibentry\undefined
		\usepackage[hidelinks,hypertexnames=false]{hyperref}
	\else
		\usepackage{bibentry}
		\makeatletter\let\saved@bibitem\@bibitem\makeatother
		\usepackage[hidelinks,hypertexnames=false]{hyperref}
		\makeatletter\let\@bibitem\saved@bibitem\makeatother
	\fi
\else
	\ifx\loadbibentry\undefined
		\relax
	\else
		\usepackage{bibentry}
	\fi
\fi

\usepackage[capitalize]{cleveref}
\makeatletter
\def\cref@getref#1#2{%
  \expandafter\let\expandafter#2\csname r@#1@cref\endcsname%
  \expandafter\expandafter\expandafter\def%
    \expandafter\expandafter\expandafter#2%
    \expandafter\expandafter\expandafter{%
      \expandafter\@firstoffive#2}}% <-------- five
\def\cpageref@getref#1#2{%
  \expandafter\let\expandafter#2\csname r@#1@cref\endcsname%
  \expandafter\expandafter\expandafter\def%
    \expandafter\expandafter\expandafter#2%
    \expandafter\expandafter\expandafter{%
      \expandafter\@secondoffive#2}}% <----------- five
      
\AtBeginDocument{%
   \def\label@noarg#1{%
    \cref@old@label{#1}%
    \@bsphack%
    \edef\@tempa{{page}{\the\c@page}}%
    \setcounter{page}{1}%
    \edef\@tempb{\thepage}%
    \expandafter\setcounter\@tempa%
    \cref@constructprefix{page}{\cref@result}%
    \protected@write\@auxout{}%
      {\string\newlabel{#1@cref}{{\cref@currentlabel}%
      {[\@tempb][\arabic{page}][\cref@result]\thepage}{}{}{}}}% <----- five
    \@esphack}%
  \def\label@optarg[#1]#2{%
    \cref@old@label{#2}%
    \@bsphack%
    \edef\@tempa{{page}{\the\c@page}}%
    \setcounter{page}{1}%
    \edef\@tempb{\thepage}%
    \expandafter\setcounter\@tempa%
    \cref@constructprefix{page}{\cref@result}%
    \protected@edef\cref@currentlabel{%
      \expandafter\cref@override@label@type%
        \cref@currentlabel\@nil{#1}}%
    \protected@write\@auxout{}%
      {\string\newlabel{#2@cref}{{\cref@currentlabel}%
      {[\@tempb][\arabic{page}][\cref@result]\thepage}{}{}{}}}% <------- five
    \@esphack}%
    }  
\makeatother
\crefname{equation}{}{}
\Crefname{equation}{}{}
\crefname{claim}{claim}{claims}
\crefname{step}{step}{steps}
\crefname{line}{line}{lines}
\crefname{condition}{condition}{conditions}
\crefname{dmath}{}{}
\crefname{dseries}{}{}
\crefname{dgroup}{}{}
\crefname{page}{page}{pages}

\crefname{Problem}{Problem}{Problems}
\crefformat{Problem}{Problem~#2#1#3}
\crefrangeformat{Problem}{Problems~#3#1#4 to~#5#2#6}

\crefname{Theorem}{Theorem}{Theorems}
\crefname{Corollary}{Corollary}{Corollaries}
\crefname{Proposition}{Proposition}{Propositions}
\crefname{Lemma}{Lemma}{Lemmas}
\crefname{Definition}{Definition}{Definitions}
\crefname{Example}{Example}{Examples}
\crefname{Assumption}{Assumption}{Assumptions}
\crefname{Remark}{Remark}{Remarks}
\crefname{Rem}{Remark}{Remarks}
\crefname{remarks}{Remarks}{Remarks}
\crefname{Appendix}{Appendix}{Appendices}
\crefname{Supplement}{Supplement}{Supplements}
\crefname{Exercise}{Exercise}{Exercises}
\crefname{TheoremA}{Theorem}{Theorems}
\crefname{CorollaryA}{Corollary}{Corollaries}
\crefname{PropositionA}{Proposition}{Propositions}
\crefname{LemmaA}{Lemma}{Lemmas}
\crefname{DefinitionA}{Definition}{Definitions}
\crefname{ExampleA}{Example}{Examples}
\crefname{RemarkA}{Remark}{Remarks}
\crefname{AssumptionA}{Assumption}{Assumptions}
\crefname{ExerciseA}{Exercise}{Exercises}

\usepackage{crossreftools}
\ifx\notloadhyperref\undefined
\else
	\relax
\fi

\usepackage{algorithm}
\usepackage{algpseudocode}

\ifx\loadbreqn\undefined
	\relax
\else
	\usepackage{breqn}
\fi

\makeatletter
\def\cleartheorem#1{%
    \expandafter\let\csname#1\endcsname\relax
    \expandafter\let\csname c@#1\endcsname\relax
}
\def\clearthms#1{ \@for\tname:=#1\do{\cleartheorem\tname} }
\makeatother

\ifx\renewtheorem\undefined
	\ifx\useTheoremCounter\undefined
		\newtheorem{Theorem}{Theorem}
		\newtheorem{Corollary}{Corollary}
		\newtheorem{Proposition}{Proposition}
		\newtheorem{Lemma}{Lemma}
	\else
		\newtheorem{Theorem}{Theorem}
		\newtheorem{Corollary}[Theorem]{Corollary}
		\newtheorem{Proposition}[Theorem]{Proposition}
	\fi

	\newtheorem{Definition}{Definition}
	
	\newtheorem{Remark}{Remark}

\fi

\theoremstyle{remark}

\theoremstyle{plain}

\newcommand{\qednew}{\nobreak \ifvmode \relax \else
		\ifdim\lastskip<1.5em \hskip-\lastskip
			\hskip1.5em plus0em minus0.5em \fi \nobreak
		\vrule height0.75em width0.5em depth0.25em\fi}



\NewDocumentCommand{\movedownsub}{e{^_}}{%
	\IfNoValueTF{#1}{%
		\IfNoValueF{#2}{^{}}% neither ^ nor _, do nothing; if no ^ but _, add ^{}
	}{%
		^{#1}% add superscript if present
	}%
	\IfNoValueF{#2}{_{#2}}% add subscript if present
}

\let\latexchi\chi
\RenewDocumentCommand{\chi}{}{\latexchi\movedownsub}

\newcommand{\calB}{\mathcal{B}}
\newcommand{\calC}{\mathcal{C}}

\newcommand{\calE}{\mathcal{E}}

\newcommand{\calL}{\mathcal{L}}

\newcommand{\calN}{\mathcal{N}}

\newcommand{\calS}{\mathcal{S}}

\newcommand{\calV}{\mathcal{V}}

\newcommand{\bA}{\mathbf{A}}

\newcommand{\bH}{\mathbf{H}}

\newcommand{\bN}{\mathbf{N}}

\newcommand{\bs}{\mathbf{s}}

\newcommand{\bU}{\mathbf{U}}

\newcommand{\bx}{\mathbf{x}}
\newcommand{\bX}{\mathbf{X}}

\newcommand{\bbN}{\mathbb{N}}

\newcommand{\bbR}{\mathbb{R}}

\newcommand{\bbZ}{\mathbb{Z}}

\DeclareSymbolFont{bsfletters}{OT1}{cmss}{bx}{n}
\DeclareSymbolFont{ssfletters}{OT1}{cmss}{m}{n}
\DeclareMathSymbol{\bsfGamma}{0}{bsfletters}{'000}
\DeclareMathSymbol{\ssfGamma}{0}{ssfletters}{'000}
\DeclareMathSymbol{\bsfDelta}{0}{bsfletters}{'001}
\DeclareMathSymbol{\ssfDelta}{0}{ssfletters}{'001}
\DeclareMathSymbol{\bsfTheta}{0}{bsfletters}{'002}
\DeclareMathSymbol{\ssfTheta}{0}{ssfletters}{'002}
\DeclareMathSymbol{\bsfLambda}{0}{bsfletters}{'003}
\DeclareMathSymbol{\ssfLambda}{0}{ssfletters}{'003}
\DeclareMathSymbol{\bsfXi}{0}{bsfletters}{'004}
\DeclareMathSymbol{\ssfXi}{0}{ssfletters}{'004}
\DeclareMathSymbol{\bsfPi}{0}{bsfletters}{'005}
\DeclareMathSymbol{\ssfPi}{0}{ssfletters}{'005}
\DeclareMathSymbol{\bsfSigma}{0}{bsfletters}{'006}
\DeclareMathSymbol{\ssfSigma}{0}{ssfletters}{'006}
\DeclareMathSymbol{\bsfUpsilon}{0}{bsfletters}{'007}
\DeclareMathSymbol{\ssfUpsilon}{0}{ssfletters}{'007}
\DeclareMathSymbol{\bsfPhi}{0}{bsfletters}{'010}
\DeclareMathSymbol{\ssfPhi}{0}{ssfletters}{'010}
\DeclareMathSymbol{\bsfPsi}{0}{bsfletters}{'011}
\DeclareMathSymbol{\ssfPsi}{0}{ssfletters}{'011}
\DeclareMathSymbol{\bsfOmega}{0}{bsfletters}{'012}
\DeclareMathSymbol{\ssfOmega}{0}{ssfletters}{'012}

\makeatletter
\newcommand*\rel@kern[1]{\kern#1\dimexpr\macc@kerna}
\newcommand*\widebar[1]{%
  \begingroup
  \def\mathaccent##1##2{%
    \rel@kern{0.8}%
    \overline{\rel@kern{-0.8}\macc@nucleus\rel@kern{0.2}}%
    \rel@kern{-0.2}%
  }%
  \macc@depth\@ne
  \let\math@bgroup\@empty \let\math@egroup\macc@set@skewchar
  \mathsurround\z@ \frozen@everymath{\mathgroup\macc@group\relax}%
  \macc@set@skewchar\relax
  \let\mathaccentV\macc@nested@a
  \macc@nested@a\relax111{#1}%
  \endgroup
}
\makeatother

\DeclareMathOperator{\var}{var}

\DeclareMathOperator{\cov}{cov}

\newcommand{\ifbcdot}[1]{\ifblank{#1}{\cdot}{#1}}

\DeclarePairedDelimiterX\abs[1]{\lvert}{\rvert}{\ifbcdot{#1}}
\DeclarePairedDelimiterX\parens[1]{(}{)}{\ifbcdot{#1}}
\DeclarePairedDelimiterX\brk[1]{[}{]}{\ifbcdot{#1}}
\DeclarePairedDelimiterX\braces[1]{\{}{\}}{\ifbcdot{#1}}
\DeclarePairedDelimiterX\angles[1]{\langle}{\rangle}{\ifblank{#1}{\cdot,\cdot}{#1}}
\DeclarePairedDelimiterX\ip[2]{\langle}{\rangle}{\ifbcdot{#1},\ifbcdot{#2}}
\DeclarePairedDelimiterX\norm[1]{\lVert}{\rVert}{\ifbcdot{#1}}
\DeclarePairedDelimiterX\ceil[1]{\lceil}{\rceil}{\ifbcdot{#1}}
\DeclarePairedDelimiterX\floor[1]{\lfloor}{\rfloor}{\ifbcdot{#1}}

\DeclareFontFamily{U}{matha}{\hyphenchar\font45}
\DeclareFontShape{U}{matha}{m}{n}{
      <5> <6> <7> <8> <9> <10> gen * matha
      <10.95> matha10 <12> <14.4> <17.28> <20.74> <24.88> matha12
      }{}
\DeclareSymbolFont{matha}{U}{matha}{m}{n}
\DeclareFontSubstitution{U}{matha}{m}{n}

\DeclareFontFamily{U}{mathx}{\hyphenchar\font45}
\DeclareFontShape{U}{mathx}{m}{n}{
      <5> <6> <7> <8> <9> <10>
      <10.95> <12> <14.4> <17.28> <20.74> <24.88>
      mathx10
      }{}
\DeclareSymbolFont{mathx}{U}{mathx}{m}{n}
\DeclareFontSubstitution{U}{mathx}{m}{n}

\DeclareMathDelimiter{\vvvert}{0}{matha}{"7E}{mathx}{"17}
\DeclarePairedDelimiterX\vertiii[1]{\vvvert}{\vvvert}{\ifbcdot{#1}}

\DeclarePairedDelimiterXPP\trace[1]{\operatorname{Tr}}{(}{)}{}{\ifbcdot{#1}} % column vector
\DeclarePairedDelimiterXPP\col[1]{\operatorname{col}}{\{}{\}}{}{\ifbcdot{#1}} % column vector
\DeclarePairedDelimiterXPP\row[1]{\operatorname{row}}{\{}{\}}{}{\ifbcdot{#1}} % row vector
\DeclarePairedDelimiterXPP\erf[1]{\operatorname{erf}}{(}{)}{}{\ifbcdot{#1}}
\DeclarePairedDelimiterXPP\erfc[1]{\operatorname{erfc}}{(}{)}{}{\ifbcdot{#1}}
\DeclarePairedDelimiterXPP\KLD[2]{D}{(}{)}{}{\ifbcdot{#1}\, \delimsize\|\, \ifbcdot{#2}} % KL divergence
\DeclarePairedDelimiterXPP\op[2]{\operatorname{#1}}{(}{)}{}{#2} % general operator

\newcommand{\ud}{\,\mathrm{d}} % for integrals like \int f(x) \ud x
\DeclarePairedDelimiterXPP\indicate[1]{{\bf 1}}{\{}{\}}{}{\ifbcdot{#1}}

\newcommand{\tc}[1]{^{(#1)}}
\NewDocumentCommand\ofrac{s m}{%
	\IfBooleanTF#1%
	{\dfrac{1}{#2}}%
	{\frac{1}{#2}}%
}
\NewDocumentCommand\ddfrac{s m m}{%
	\IfBooleanTF#1%
	{\dfrac{\mathrm{d} {#2}}{\mathrm{d} {#3}}}%
	{\frac{\mathrm{d} {#2}}{\mathrm{d} {#3}}}%
}
\NewDocumentCommand\ppfrac{s m m}{%
	\IfBooleanTF#1%
	{\dfrac{\partial {#2}}{\partial {#3}}}%
	{\frac{\partial {#2}}{\partial {#3}}}%
}

\newcommand{\setgiven}{:}
\providecommand\given{}
\DeclarePairedDelimiterX\Set[2]\{\}{%
	\if#1:%
		\renewcommand\given{\SetSymbol{:}}%
	\else%
		\renewcommand\given{\SetSymbol[\delimsize]{#1}}%
	\fi%
#2
}

\NewDocumentCommand\set{s O{\setgiven} m}{%
	\IfBooleanTF#1%
	{\Set*{#2}{#3}}%
	{\Set{#2}{#3}}%
}

\NewDocumentCommand{\evalat}{ s O{\big} m e{_^} }{%
\IfBooleanTF{#1}%
{\left. #3 \right|}{#3#2|}%
\IfValueT{#4}{_{#4}}%
\IfValueT{#5}{^{#5}}%
}

\providecommand\given{}
\DeclarePairedDelimiterXPP\cprob[1]{}(){}{
\renewcommand\given{\nonscript\,\delimsize\vert\allowbreak\nonscript\,\mathopen{}}%
#1%
}
\DeclarePairedDelimiterXPP\cexp[1]{}[]{}{
\renewcommand\given{\nonscript\,\delimsize\vert\allowbreak\nonscript\,\mathopen{}}%
#1%
}

\DeclareDocumentCommand \P { s e{_^} d() g } {%
	\mathbb{P}%
	\IfBooleanTF{#1}%
		{
			\IfValueT{#2}{_{#2}}%
			\IfValueT{#3}{^{#3}}%
			\IfValueTF{#5}{\cprob{#4 \given #5}}{\IfValueT{#4}{\cprob{#4}}}%
		}%
		{
			\IfValueT{#2}{_{#2}}%
			\IfValueT{#3}{^{#3}}%
			\IfValueTF{#5}{\cprob*{#4 \given #5}}{\IfValueT{#4}{\cprob*{#4}}}%
		}%
}

\DeclareDocumentCommand \E { s e{_^} o g } {%
	\mathbb{E}%
	\IfBooleanTF{#1}%
		{
			\IfValueT{#2}{_{#2}}%
			\IfValueT{#3}{^{#3}}%
			\IfValueTF{#5}{\cexp{#4 \given #5}}{\IfValueT{#4}{\cexp{#4}}}%
		}%
		{
			\IfValueT{#2}{_{#2}}%
			\IfValueT{#3}{^{#3}}%	
			\IfValueTF{#5}{\cexp*{#4 \given #5}}{\IfValueT{#4}{\cexp*{#4}}}%		
		}%
}

\DeclareDocumentCommand \Var { s e{_^} d() g } {%
	\var%
	\IfBooleanTF{#1}%
		{
			\IfValueT{#2}{_{#2}}%
			\IfValueT{#3}{^{#3}}%
			\IfValueTF{#5}{\cprob{#4 \given #5}}{\IfValueT{#4}{\cprob{#4}}}%
		}%
		{
			\IfValueT{#2}{_{#2}}%
			\IfValueT{#3}{^{#3}}%	
			\IfValueTF{#5}{\cprob*{#4 \given #5}}{\IfValueT{#4}{\cprob*{#4}}}%		
		}%
}

\DeclareDocumentCommand \Cov { s e{_^} d() g } {%
	\cov%
	\IfBooleanTF{#1}%
		{
			\IfValueT{#2}{_{#2}}%
			\IfValueT{#3}{^{#3}}%
			\IfValueTF{#5}{\cprob{#4 \given #5}}{\IfValueT{#4}{\cprob{#4}}}%
		}%
		{
			\IfValueT{#2}{_{#2}}%
			\IfValueT{#3}{^{#3}}%	
			\IfValueTF{#5}{\cprob*{#4 \given #5}}{\IfValueT{#4}{\cprob*{#4}}}%		
		}%
}

\ExplSyntaxOn
\NewDocumentCommand \dist {m o o} {%
\mathrm{#1}\left(%
	\IfValueT{#3}{%
		\tl_if_blank:nTF{ #3 }{\cdot\, \middle|\, }{#3\, \middle|\, }%
	}
	\IfValueT{#2}{#2}%
\right)%
}
\ExplSyntaxOff

\NewDocumentCommand {\cbrace} {t+ D[]{black} D(){\widthof{#5}} m m } {%
	\begingroup%
		\color{#2}
		\IfBooleanTF{#1}{%
			\overbrace{#4}^%
		}{
			\underbrace{#4}_%
		}%
		{\parbox[c]{#3}{\centering\footnotesize{#5}}}%
	\endgroup% 
}

\let\oldforall\forall
\renewcommand{\forall}{\oldforall \, }

\let\oldexist\exists
\renewcommand{\exists}{\oldexist \, }

\makeatletter

\newcommand{\rankcolor}[2]{%
	\expandafter\renewcommand\csname #1\endcsname[1]{%
		\ifblank{##1}{%
			{\color{#2} \textbf{#2}}%
		}{%
			\ifmmode
				\textcolor{#2}{\bm{##1}}%
			\else%
				{\color{#2} \textbf{##1}}%
			\fi	
		}%
	}
}

\rankcolor{first}{red}
\rankcolor{second}{blue}
\rankcolor{third}{cyan}
\makeatother

\graphicspath{{./Figures/}{./figures/}}
\DeclareDocumentCommand{\includeCroppedPdf}{ o O{./Figures/} m }{
	\IfFileExists{#2#3-crop.pdf}{}{%
		\immediate\write18{pdfcrop #2#3.pdf #2#3-crop.pdf}}%
	\includegraphics[#1]{#2#3-crop.pdf}
}

\makeatletter
\newcommand*{\addFileDependency}[1]{% argument=file name and extension
  \typeout{(#1)}
  \@addtofilelist{#1}
  \IfFileExists{#1}{}{\typeout{No file #1.}}
}
\makeatother

\definecolor{gray90}{gray}{0.9}
\def\colorlist{red,blue,brown,cyan,darkgray,gray,lightgray,green,lime,magenta,olive,orange,pink,purple,teal,violet,white,yellow}

\makeatletter
\def\startcomment{[}
\ifx\nohighlights\undefined
	\newcommand{\createcolor}[1]{%
			\expandafter\newcommand\csname #1\endcsname[1]{{\color{#1} ##1}}%
	}
	\newcommand{\msout}[1]{\text{\color{green} \st{\ensuremath{#1}}}}
	\newcommand{\del}[1]{{\color{green}\ifmmode \msout{#1}\else\st{#1}\fi}}
\else
	\newcommand{\createcolor}[1]{%
			\expandafter\newcommand\csname #1\endcsname[1]{%
				\noexpandarg%
				\StrChar{##1}{1}[\firstletter]%
				\if\firstletter\startcomment%
					\relax
				\else%
					##1
				\fi
			}%
	}
	\newcommand{\msout}[1]{}
	\newcommand{\del}[1]{}
\fi

\def\@tempa#1,{%
    \ifx\relax#1\relax\else
        \createcolor{#1}%
        \expandafter\@tempa
    \fi
}
\expandafter\@tempa\colorlist,\relax,
\makeatother

\newcommand{\hhide}[1]{}
\ifx\diagnoselabel\undefined
	\relax
\else
	\makeatletter
	\def\@testdef #1#2#3{%
		\def\reserved@a{#3}\expandafter \ifx \csname #1@#2\endcsname
			\reserved@a  \else
			\typeout{^^Jlabel #2 changed:^^J%
				\meaning\reserved@a^^J%
				\expandafter\meaning\csname #1@#2\endcsname^^J}%
			\@tempswatrue \fi}
	\makeatother
\fi